\documentclass[12pt,a4paper]{article}

\usepackage[english]{babel}
\usepackage{amsthm}

\usepackage[letterpaper,top=2cm,bottom=2cm,left=3cm,right=3cm,marginparwidth=1.75cm]{geometry}

\usepackage{amsmath}
\usepackage{graphicx}
\usepackage{ytableau}
\usepackage{booktabs}
\usepackage[colorlinks=true, allcolors=blue]{hyperref}
\usepackage{enumerate}
\usepackage[shortlabels]{enumitem}

\usepackage{tikz,fullpage}
\usetikzlibrary{arrows,%
                petri,%
                topaths,%
                positioning}
\usepackage{tkz-berge}
\usepackage[position=top]{subfig}
\usepackage{pgf}

\usetikzlibrary{arrows,automata}
\usepackage{inputenc}

\newtheorem{Theorem}{Theorem}[section]
\newtheorem{Proposition}[Theorem]{Proposition}
\newtheorem{Lemma}[Theorem]{Lemma}

\newtheorem*{Example}{Example}

\theoremstyle{definition}
\newtheorem{Definition}[Theorem]{Definition}
\newtheorem{Observation}[Theorem]{Observation}

\newcommand{\C}{\mathcal C}

\title{Counting spinal phylogenetic networks}
\author{Andrew Francis and Michael Hendriksen}

\begin{document}
\maketitle

\begin{abstract}
Phylogenetic networks are an important way to represent evolutionary histories that involve reticulations such as hybridization or horizontal gene transfer, yet fundamental questions such as how many networks there are that satisfy certain properties are very difficult.  A new way to encode a large class of networks, using expanding covers, may provide a way to approach such problems.  Expanding covers encode a large class of phylogenetic networks, called labellable networks. This class does not include all networks, but does include many familiar classes, including orchard, normal, tree-child and tree-sibling networks. As expanding covers are a combinatorial structure, it is possible that they can be used as a tool for counting such classes for a fixed number of leaves and reticulations, for which, in many cases, a closed formula has not yet been found. More recently, a new class of networks was introduced, called spinal networks, which are analogous to caterpillar trees for phylogenetic trees and can be fully described using covers. In the present article, we describe a method for counting networks that are both spinal and belong to some more familiar class, with the hope that these form a base case from which to attack the more general classes.
\end{abstract}

\section{Introduction}
Phylogenetic trees and networks are important ways to represent the evolutionary relationships between a set of taxa.  While trees describe histories that model evolution together with speciation, networks are able to include other ``reticulate'' processes such as horizontal gene transfer, hybridization, or endosymbiosis.  The capacity to model these processes comes at the price of complexity, and difficulties in inferring networks from the data, which are typically genetic sequences from extant taxa at the leaves of the network. 

In particular, while there are finitely many trees on a given number of leaves (exactly $(2n-3)!!$ rooted binary trees on $n$ leaves), without further qualification there are infinitely many rooted binary phylogenetic networks.   This is one reason why many classes of phylogenetic network have been defined, in attempts to find restrictions that make the problem more mathematically tractable, and hopefully also biologically plausible.  A recent summary of many currently described classes can be found in~\cite{kong2022classes}.  New classes, however, continue to be defined, and in the present paper we will show that one such new class, the \emph{spinal phylogenetic networks}, can be precisely enumerated.  

The class of spinal networks was defined in the context of the study of bijections between phylogenetic networks and covers of a finite set~\cite{francis2023phylogenetic}.  They can be defined as those networks for which there is a path from a leaf to the root that includes all non-leaf vertices (a ``spine'').  The \emph{trees} that satisfy this property are called ``caterpillar'' trees, and are an important class of trees for a number of reasons (for instance, they are those trees with a single cherry, and are the least ``balanced''~\cite{fischer2023tree}).  Consequently, studying combinatorial properties of spinal networks may provide a useful edge-case for the study of phylogenetic networks in more generality.

We will show in Section~\ref{s:counting.spinal} that the spinal networks can be precisely counted, along with a number of subclasses of spinal networks that have additional properties.  Of particular interest are the spinal networks that are also \emph{tree-child}.  The tree-child networks~\cite{cardona2008comparison} (or TCNs) are an important class because they have both properties that make them mathematically workable, as well as properties that may make them possible to reconstruct from data at the leaves.

A considerable amount of recent work has explored the problem of enumeration of tree-child networks. For instance, Fuchs et al. introduced a method based on generating functions to count TCNs with few reticulation vertices \cite{fuchs2019counting}, and Cardona and Zhang introduced an algorithm for enumeration of TCNs, allowing TCNs with up to $8$ taxa and up to $7$ reticulations to be counted \cite{cardona2020counting}. 
A range of approaches have also been taken to enumerate classes of phylogenetic networks asymptotically, including tree-child networks \cite{mcdiarmid2015counting,fuchs2019counting,fuchs2021asymptotic}. 
Despite this, no closed formula is known for counting TCNs on $n$ taxa with $k$ reticulations in general.

It is to be hoped that counting these classes of spinal networks will provide a valuable base from which to approach harder open problems, such as counting tree-child networks more generally.  For instance, one conjecture is that the number $|TC_{n,k}|$ of binary tree-child networks with $k$ reticulations and $n$ leaves is given by the recursion~\cite{pons2021combinatorial,fuchs2021asymptotic}:
\begin{equation}\label{e:TCN.recursion}
(n - k)|TC_{n,k}|=(n + 1 - k)(n - k) |TC_{n,k-1}| + n (2n + k - 3) |TC_{n-1,k}|.    
\end{equation}
The spinal networks that are tree-child may be able to play a role as a base case from which to identify a recursion for the number of tree-child networks.

This paper will count spinal networks as well as some subclasses of them, using the bijection between labellable phylogenetic networks and expanding covers of a finite set~\cite{francis2023labellable}. These subclasses of spinal networks include tree-child networks, binary networks, stack-free networks and a novel class termed fully tree-sibling networks.

\section{Preliminaries}

This article defines phylogenetic networks as follows.

\begin{Definition}
    A (rooted) phylogenetic network on leaf-set $X$, with $|X|=n$, is a directed acyclic graph whose vertex set satisfies the following:
    \begin{itemize}
        \item There is one vertex of in-degree 0, called the root;
        \item There are $n$ vertices of in-degree 1 and out-degree 0, called leaves and labelled by elements of $X$; and
        \item All other vertices can have any in-degree or out-degree $\ge 1$.
    \end{itemize}
\end{Definition}

Note that this definition is more broad than many definitions of phylogenetic network, in that it permits so-called \emph{degenerate} vertices, which have either in-degree and out-degree 1, or in-degree and out-degree strictly greater than 1.  Networks containing a degenerate vertex are called degenerate.
A network is considered \textit{binary} if all internal vertices have total degree exactly 3, and the root must have out-degree 2.  The recursive formula conjectured to count tree-child networks (Eq.~\eqref{e:TCN.recursion}) is for binary non-degenerate networks.

A vertex that has in-degree less than 2 is called a \emph{tree vertex} (including the root and the leaves), while vertices of in-degree greater than or equal to 2 are called \emph{reticulation vertices}.

In this paper we will be making use of the bijection between the class of \emph{labellable} networks, and ``expanding covers'' of finite sets~\cite{francis2023labellable}.  The labellable networks are a large class that contains all the classes discussed below, and many commonly studied classes can themselves be described in terms of a bijection with expanding covers that satisfy certain additional conditions~\cite{francis2023phylogenetic}.  

Given a phylogenetic network whose leaves are labelled by an ordered set $[n]:=\{1,\dots,n\}$, there is an algorithm for labelling the remaining non-root vertices that proceeds by identifying the unlabelled vertex with the least set of children, according to the lexicographic order on sets:
\[A \prec B \text{ if } A \subseteq B \text{ or } \min (A \backslash (A \cap B)) < \min (B \backslash (A \cap B)).\]
This algorithm generalises a similar one for phylogenetic trees~\cite{erdos1989applications}.

A phylogenetic network is said to be \emph{labellable} if this algorithm is well-defined.  Labellable networks can be characterised in terms of certain forbidden substructures, but for our purposes the important characterisation is that they are in bijection with the set of \emph{expanding covers}:

\begin{Definition}[\cite{francis2023labellable}]\label{d:expanding.cover}
A cover $\C$ of $[m]$ is expanding if, for $n = m - |\C| + 1$, it satisfies:
\begin{enumerate}  
    \item No element of $[n]$ appears more than once, and
    \item For $i = 1, \dots, |\C|$, the cover contains at least $i$ subsets of $[n + i - 1]$.
\end{enumerate}
\end{Definition}

In fact, we have the following theorem.

\begin{Theorem}[\cite{francis2023labellable}, Theorem 4.4]
    The set of labellable phylogenetic networks is in bijection with the set of expanding covers.
\end{Theorem}

The labelling algorithm and Definition~\ref{d:expanding.cover} naturally lead to a \textit{labelling order}, defined in the following way.

The set giving rise to the vertex label $i$ is determined by its position in an ordering of the sets in $\C$. This ordering will use $\prec$ and is defined by the following algorithm:
\begin{enumerate}
\item For $i = 1, \dots , |\C|$, do:
\begin{enumerate}
\item $C_i$ is the minimal set in $(\C,\prec)$ contained in $[n + i - 1]$.
\item Set $\C = \C \backslash C_i$ .
\end{enumerate}
\item Output sequence $C_1, \dots ,C_{|\C|}$.
\end{enumerate}

We can now move to the main topic of this paper, which is the class of \emph{spinal networks}. We define a \textit{spine} in a network to be a path from a leaf to the root that traverses all non-leaf vertices, and we call a network \textit{spinal} if it has a spine.

\begin{Theorem}[\cite{francis2023phylogenetic}, Theorem 8.1]
    A network is spinal if and only if its cover $\C$ has exactly $i$ subsets of $[n + i - 1]$, for each $i = 1, \dots, |\C|$.
\end{Theorem} 

Compare this characterisation to Definition \ref{d:expanding.cover} - specifically that the words ``at least'' are replaced by ``exactly''. We note that the only spinal networks that are also phylogenetic trees are precisely the set of caterpillar trees. For this reason, spinal networks can be seen as a generalisation of caterpillar trees.

\section{Counting Classes of Spinal Networks}
\label{s:counting.spinal}

We will now proceed to leverage the expanding cover interpretation of labellable phylogenetic networks, together with the naturally restrictive characterisation of spinal networks, to enumerate various classes of spinal networks. With the exception of the class of all spinal networks, which we will deal with first, these enumerations will follow a reasonably similar structure - first characterising the vertices in a spine of a spinal network of some particular class, then creating a surjection based on this characterisation into a nice set-theoretic class, and then modifying this result to fully count the relevant subclass of spinal networks. We will give a more complete overview of this structure in Section \ref{s:proofsketch}.

In Section \ref{s:all}, phylogenetic networks are permitted to be degenerate and non-binary, but in subsequent sections every network will be binary, non-degenerate, and spinal. We begin by noting some helpful lemmas.

For some class of networks $R$, define the related class of \textit{network shapes in} $R$, denoted $S(R)$, to be the set of equivalence classes of $R$ under the equivalence relation $\sim_S$, where $N_1 \sim_S N_2$ are equivalent if there exists a permutation on the labels of the leaves of $N_1$ for which we obtain the network $N_2$. 

We will also need the notion of \textit{cherry reduction} for a binary spinal network shape $N$. If a vertex has only one parent, denote the parent of a vertex $x$ by $p(x)$. A pair of leaves $(x,y)$, who by definition only have one parent each,  is a \textit{cherry} in a network if $x$ and $y$ have the same parent $p(x)=p(y)$. Reducing a cherry $(x, y)$ in a network is the action of deleting $x$ and replacing the two edges $(p(p(x)), p(x))$ and $(p(x), y)$ with a single edge $(p(p(x)), y)$. Note that $p(p(x))$ is well-defined, as $N$ is a binary network, so as $p(x)$ has two child vertices it can only have one parent vertex.

We also define a \textit{re-cherrying} operation on a binary spinal network, in which one of the (up to) two leaf edges furthest from the root is subdivided, and a new leaf is attached to the resulting degree-$2$ vertex. Note that in the context of binary spinal networks, the only possible cherry that can appear would be the children of the internal vertex furthest from the root, so cherry reduction and re-cherrying are inverse operations.

\begin{Lemma}
\label{l:shapetonet}
    Let $R$ be a class of binary spinal networks that is closed under permutations of leaf labels, cherry reductions {and re-cherrying}. If $R_{n,k}$ denotes the subclass of $R$ on $n$ leaves with $k$ reticulations, then
    \[|R_{n,k}|=n!|S(R_{n,k})|-\frac{n!}{2}|S(R_{n-1,k})|.\]
\end{Lemma}

\begin{proof}
    Note that for any fixed spinal network $N$, there exists a unique path from a leaf to the root up to choice of the initial leaf if the first internal vertex of the path has two leaf children. Therefore, for each leaf $\ell$ that is a child of the $i$-th internal vertex, any permutation on the labels of the leaves that does not have $\ell$ as a fixed point will result in a different network that has the same network shape. Denote the subclass of $S(R_{n,k})$ with $i$ children of the first vertex by $L_i$, noting that $|S(R_{n,k})|=|L_1|+|L_2|$. It follows that if $N$ has only one child of the first internal vertex (that is, $N$ has no cherries), the equivalence class of network shapes in $L_1$ containing $N$ will contain $n!$ networks (because each permutation of the $n$ leaves will give a distinct network). However, if $N$ has two children of the first internal vertex, the equivalence class of network shapes in $L_2$ will contain $n!/2$ networks (as a consequence of the fact that the permutation swapping these labels of these two children does not result in a new network). It follows that
    \[|R_{n,k}|=n!|L_1|+\frac{n!}{2}|L_2|=n!(|S(R_{n,k})|-|L_2|)+\frac{n!}{2}|L_2|=n!|S(R_{n,k})|-\frac{n!}{2}|L_2|.\]

    We finally note that there exists a straightforward bijection between $L_2$ and network shapes in $S(R_{n-1,k})$ obtained by cherry reduction in the forward direction and re-cherrying in the inverse. 
    The result follows.
\end{proof}

\subsection{The class of all spinal networks}
\label{s:all}

Denote the cardinality of the class of spinal networks on $n$ leaves with covers of size $|\C|$ by $S_{n,|\C|}$. Note that, in contrast to all other results in this manuscript, we are not restricting these networks to be binary.

\begin{Proposition}
    \label{p:allspinal}
    $S_{n,|\C|} = 2^{\binom{|\C|-2}{2}}(|\C|^n - (|\C|-1)^n)$.
\end{Proposition}

\begin{proof}
    Note first that by virtue of being a spinal network, the structure of the expanding cover under the labelling order is very restricted. In particular, as noted in Theorem 8.1 of \cite{francis2023phylogenetic}, for each $i=2,\dots,|\C|$, the set $C_i$ must contain $n+i-1$ (note, it is stated in \cite{francis2023phylogenetic} that $C_1$ must also have this property, but this is not necessary). This means it is not necessary to consider the labelling order explicitly, as the position of a set in the order is made clear by considering whether the largest element in the set is a leaf (in which case the set is $C_1$) or the element $n+i-1$ for $i>1$, in which case it is $C_i$.

    We now assign the leaves. The leaves may be assigned to any set in any order, but each must be assigned to a single set, and with the caveat that $C_1$ must be non-empty. There are therefore 
    \begin{equation}
    \label{e:leaves}
        |\C|^n - (|\C|-1)^n
    \end{equation}
    options for doing this (all options, minus those for which $C_1$ is empty). 
    
    Then, for each $i \in 2,\dots,|\C|$ each element $n+i-1$ must be contained in set $C_i$, and for $j=i+1,\dots,|\C|$ the element $n+i-1$ either is or is not contained in set $C_j$, for a total of $2^{|\C|-i}$ possibilities. Multiplying over all possible $i$, we get

    \begin{equation}
    \label{e:internal}
        \prod_{i=2}^{|\C|} 2^{|\C|-i} = 2^{\sum_{j=0}^{|\C|-2} j} = 2^{\binom{|\C|-2}{2}}.
    \end{equation}

    Multiplying expressions \eqref{e:leaves} and \eqref{e:internal} together gives the result in the proposition.
\end{proof}

\subsection{A brief overview of some subsequent proofs}
\label{s:proofsketch}

Results in each of the following $3$ sections (Sections \ref{s:caterpillar}, \ref{s:stackfree} and \ref{s:fullytreesibling}) follow a similar method of proof, which we will describe informally for ease of following the proofs.

In each case, we will specifically count the binary network shapes with no parallel edges that belong to some given class (tree-child, stack-free or fully tree-sibling) and have $n$ leaves and $k$ reticulations, which here we will denote $\mathcal{TS}(n,k)$. We will refer to the path along the spine as the \textit{spinal path}, which will start at a leaf and end at the root. We note here that this is uniquely defined up to the choice (between up to $2$ options) of the initial leaf. Then, each internal vertex $v$ along the spinal path may be labelled according to which of the three mutually exclusive properties it has: 

\begin{enumerate}
    \item $v$ has a (L)eaf child that is not on the spine ($L$),
    \item $v$ is the $i$-th (R)eticulation vertex counting along the spinal path from the leaf ($R_i$), or
    \item $v$ is the (P)arent of the $i$-th reticulation vertex, via an edge not in the spinal path ($P_i$).
\end{enumerate} 

We can therefore encode each network shape in $\mathcal{TS}(n,k)$ as a finite sequence of the letters $\{L,R_1,\dots,R_k,P_1,\dots,P_k\}$ containing $n-1$ $L$'s and $1$ each of the remaining letters, with certain restrictions on the order in the sequence. In particular, all classes must obey the following three rules:
\begin{enumerate}[start=1,label={(\bfseries R\arabic*):}]
    \item If $j>i$, then $R_j$ must appear subsequently to $R_i$; 
    \item $P_i$ must appear subsequently to $R_i$; and
    \item $P_i$ must not immediately follow $R_i$.
\end{enumerate} 

We note that R3 follows from the requirement that our networks have no parallel edges. We will later briefly consider networks with parallel edges in Section \ref{s:parallel}, but will not follow the present proof structure.

The specific network class will then impose additional restrictions - for instance, in a stack-free network a reticulation vertex must not be a child of another reticulation vertex, which in this encoding would require that there is no subsequence of the form $R_iR_j$. We note that while we consider $3$ possible additional rules in the present article, it may be interesting to consider other additional rules or variations on those we present, and consider what class of binary spinal networks these may correspond to. This may be helpful in future exploration of binary spinal networks.

The next step involves mapping these sequences to some other known set for which the enumeration is already known. Finally, we then apply Lemma \ref{l:shapetonet} to count the number of networks, rather than just network shapes.

\subsection{The class of binary spinal tree-child networks}
\label{s:caterpillar}

In this Section, all networks are assumed to be non-degenerate. That is, every non-root internal vertex has out-degree $1$ or in-degree $1$, but not both.

\begin{Definition}[\cite{cardona2008comparison}]
    \emph{Tree-child} networks are phylogenetic networks for which every non-leaf vertex has a child that is a tree vertex. 
\end{Definition}

\begin{Theorem}[\cite{francis2023phylogenetic}, Theorem 4.1]
\label{t:tcn}
    Tree-child networks are in bijection with expanding covers for which each set contains an integer that appears exactly once in the cover.
\end{Theorem}

Let $N$ be a binary spinal TCN with associated expanding cover $\C$. As mentioned in the proof of Proposition \ref{p:allspinal}, since $N$ is spinal, there exists a simple way to determine the position of a set in the labelling order by considering whether the largest element in the set is a leaf (in which case the set is $C_1$) or the element $n+i-1$ for $i>1$, in which case it is $C_i$. Additionally, the characterisation of covers of non-degenerate networks given in \cite{francis2023labellable} can become a characterisation of binary networks with a small modification which we present here:

\begin{Theorem}
\label{t:binarycover}
Let $N$ be a labellable phylogenetic network with cover $\C$ in labelling order $C_1,\dots,C_{|\C|}$. Then $N$ is binary if and only if all sets in $\C$ have size $1$ or $2$, and, for each $i \ge 1$:
    \begin{itemize}
        \item If $|C_i | = 2$, then $(n + i )$ appears once in $\C$; and
        \item If $|C_i| = 1$, then $(n + i )$ appears twice in $\C$.
    \end{itemize}
\end{Theorem}

Together with the spinal condition, this tightly limits our options for constructing an expanding cover for a binary spinal TCN. For instance, if $|C_i | = 2$, then $(n + i )$ appears once in $\C$, and by the spinal condition we know that it precisely must appear in $C_{i+1}$. In fact, we can characterise the sets appearing in an expanding cover for a binary spinal TCN even further, as follows. 

\begin{Theorem}
\label{t:spinaltcn}
Let $N$ be a binary spinal TCN with expanding cover $\C$. hen a set $C_i\in\C$ falls into one of the following three mutually exclusive categories:

\begin{enumerate}
    \item It is part of a consecutive pair of sets $C_i$, $C_{i+1}$ such that $C_{i+1}$ contains $n+i$ and a leaf, and either $C_{i}=\{n+i-1\}$ or $i=1$ and $C_i$ contains a singleton leaf; or
    \item $C_i$ contains $n+i-1$ and a leaf and $C_{i-1} \ne \{n+i-2\}$, or $i=1$ and $C_1$ is two leaves; or
    \item $C_i$ uniquely contains $n+i-1$ and some element $n+j$ for $0<j<i-1$ so that $C_j$ is the first element of a pair in Category 1.
\end{enumerate}

Indeed, for each pair of sets in Category $1$, there is a corresponding set in category $3$.
\end{Theorem}

\begin{proof}
We first show that if there exists some set $C_i$ that is a singleton (necessarily $\{n+i-1\}$ or $i=1$ and $C_1$ contains a singleton leaf), then $C_{i+1}$ contains $n+i-1$ and a leaf. Since $|C_i| = 1$, then $n + i$ appears twice in $\C$, because $N$ is binary. As $N$ is spinal, $C_{i+1}$ must be one of the two sets that contain $n+i$. However, as $N$ is a TCN, by Theorem \ref{t:tcn}, $C_{i+1}$ must contain an integer that appears only once in the cover, so $C_{i+1}$ cannot be a singleton. Additionally, considering the other element of $C_{i+1}$, it cannot be $n+j-1$ for any $2<j<i+1$, as by the spinal property these are certainly already contained in $C_j$. It follows that the other element of $C_{i+1}$ is a leaf.

We now suppose that $|C_i|=2$ and is not the second element of a pair in Category 1. As $N$ is spinal, $C_i$ contains $(n+i-1)$ or $i=1$ and $C_1$ is two leaves, so now we consider the other element. If it is a leaf, then $C_i$ is in Category $2$. To prove the statement, it therefore suffices to show if $C_i$ is not in Category 2, then the other element, which must be $n+j$ for $0<j<i-1$ as it is not a leaf, has the property that $C_j$ is the first element of a pair in Category 1. As $N$ is a spinal network and $n+j$ is not a leaf, $n+j$ must appear in some other set in $\C$ that precedes $C_i$ in the labelling order. This implies that $C_i$ uniquely contains $n+i-1$, as $N$ is a TCN. Furthermore, by Theorem \ref{t:binarycover}, as $n+j$ appears twice in $\C$, we know $|C_j|=1$. By the first part of this lemma, as $C_j$ is a singleton it must be the first element of a pair in Category 1. 

Finally, observe that if we have some pair in Category 1, say $C_j,C_{j+1}$, then $n+j$ is contained in $C_{j+1}$ due the spinal condition, and due to Theorem \ref{t:binarycover} is repeated in some subsequent set $C_i$. In $C_i$ there must be a unique element, implying that the remaining element, necessarily $n+i-1$, is unique. This is exactly the conditions required for a set in Category 3. The theorem follows. 
\end{proof}

Define the following classes of interior vertices.  

\begin{Definition}
   Let $v$ be an interior vertex or root vertex of a binary spinal network. Each such vertex will be an endpoint of one edge that is not part of the spine, leading to some vertex $w$, which we term the \textit{non-spinal adjacent vertex} of $v$. 
   \begin{itemize}
       \item If $w$ is a leaf, we call $v$ an $L$-vertex;
       \item if $w$ is a vertex of the spinal path prior to $v$ we call $v$ a $P$-vertex; and 
       \item if $w$ is a vertex of the spinal path subsequent to $v$ we call $v$ an $R$-vertex. 
   \end{itemize}
   
\end{Definition}

Note that the non-spinal adjacent vertex cannot be immediately before or after $v$ along the spinal path, as we do not permit parallel edges. 
With these definitions, we can equivalently re-state Theorem \ref{t:spinaltcn} as follows:

\begin{Theorem}
\label{t:spinaltcn.RLP}
Let $N$ be a binary spinal TCN with expanding cover $\C$. Then a set $C_i\in\C$ falls into one of the following three mutually exclusive categories:

\begin{enumerate}
    \item It is part of a consecutive pair of sets $C_i,C_{i+1}$ such that $C_i$ corresponds to an $R$-vertex and $C_{i+1}$  corresponds to an $L$-vertex; or
    \item $C_i$ corresponds to an $L$-vertex and $C_{i-1}$ does not correspond to an $R$-vertex; or
    \item $C_i$ corresponds to a $P$-vertex that has an edge to the first element of a pair in Category 1.
\end{enumerate}
Indeed, for each pair of sets in Category $1$, there is a corresponding set in Category $3$.
\end{Theorem}

\begin{Observation}
    Considering the rules of the finite sequences in the proof sketch in Section \ref{s:proofsketch}, the characterisation presented in Theorem \ref{t:spinaltcn.RLP} is equivalent to adding a new rule (the first of three we will introduce): \textbf{(R4a)} $R_i$ must always be followed immediately by $L$. 
    
    Our enumeration proof will therefore count how many ways to arrange sub-subsequences of the forms $R_iL,L,P_i$ in which there are $k$ of the form $R_iL$, $n-1-k$ of the form $L$ and $k$ of the form $P_i$, to make our subsequence, while still enforcing the rules mentioned in Section \ref{s:proofsketch}. This will be achieved by mapping to an appropriate set.
\end{Observation}

Denote the set of binary spinal TCNs on $n$ leaves with $k$ reticulations by $\mathcal{STC}_{n,k}$, and the set of expanding covers of such networks by $C(\mathcal{STC}_{n,k})$. Denote the set of partitions of $n+k$ into $n-k$ sets of size $1$ and $k$ sets of size $2$ by $\mathcal{B}_{n,k}$. The sequence $|\mathcal{B}_{n-1,k}|$ is well-known, appearing off-by-one as sequence A001498 \cite{oeis}, the coefficients of the Bessel polynomial, and having exact formula

    \begin{equation}
    \label{e:bessel}
        |\mathcal{B}_{n-1,k}| = \frac{(n-1+k)!}{2^k(n-1-k)!k!}.
    \end{equation}

\begin{Theorem}
\label{t:binspinTCN}
Let $\mathcal{STC}_{n,k}$ be the set of binary spinal TCNs on $n$ leaves with $k$ reticulations. Then, if $n>1$,
\begin{equation}\label{e:STCnk.count}
|\mathcal{STC}_{n,k}|=n!|B_{n-1,k}|-\frac{n!}{2}|B_{n-2,k}|,
\end{equation}
where $|\mathcal{STC}_{1,k}|=1$ and $|B_{0,k}|=0$. Equivalently,
\begin{equation}\label{e:STCnk.countdirect}
|\mathcal{STC}_{n,k}|=\frac{n!(n-2+k)!(n-1+3k)}{2^{k+1}k!(n-1-k)!}.
\end{equation}

\end{Theorem} 

\begin{proof}
     We intend to find a surjection $\varphi$ between $C(\mathcal{STC}_{n,k})$ and $\mathcal{B}_{n-1,k}$ that will allow us to enumerate $\mathcal{STC}_{n,k}$.

     Suppose we have a binary spinal TCN. By definition, it contains a spine, which we recall to be a path from a leaf to the root that traverses all non-leaf vertices, which is necessarily unique up to the (one or two) choices of initial leaf. Ignoring the initial leaf, label the internal vertices/pairs of internal vertices sequentially according to the path, except at vertices in Category 1 of \ref{t:spinaltcn}, in which the label is repeated, and call the labelling $\mathcal{L}$. An example of this labelling may be seen in Figure \ref{fig:tcn}, with the type of circle indicating which category each vertex belongs to (plain for category $1$, bold for $2$ and dashed for $3$). In a binary network on $n$ leaves with $k$ reticulations, there are $n-1+2k$ internal vertices, and so $n-1+k$ total labels in $\mathcal{L}$. The surjection is then achieved by mapping a network to the element $P \in \mathcal{B}_{n-1,k}$ for which the Category 2 vertices correspond to singletons containing their label, and the linked Category 1 pairs with their corresponding Category 3 vertices (viz. the final sentence of the statement of Theorem \ref{t:spinaltcn}) correspond to the parts of size two containing their labels.

     However, note that this bijection was determined entirely independently of leaf labels, and thus the image of an element under $\varphi$ is only unique up to tree shape. It follows that $\mathcal B_{n-1,k}$ counts exactly the network shapes of elements of $\mathcal{STC}_{n,k}$, and to find $|\mathcal{STC}_{n,k}|$ we may appeal to Lemma \ref{l:shapetonet}, resulting in the statement of the theorem. We can find the equivalent direct formulation by substituting in Equation \eqref{e:bessel} appropriately.
\end{proof}

\begin{figure}[ht]
    \centering
    \begin{tikzpicture}[dot/.style={circle,fill=white,radius=3pt,inner sep = 3pt}] 
        \draw (0, 0) -- (7,7);
        \draw (2, 2) -- (3,1);
        \draw (5,5) -- (6,4);
        \draw (7,7) -- (8,6);
        \draw [bend left=45] (1,1) to (3,3);
        \draw [bend left=45] (4,4) to (6,6);
        
        \node[draw,dot] at (0, 0)   (a) {};
        \node[draw,dot] at (1, 1)   {1};
        \node[draw,dot] at (2, 2)   (c) {1};
        \node[draw,dot,dashed] at (3, 3)   (d) {2};
        \node[draw,dot] at (4, 4)   (e) {3};
        \node[draw,dot] at (5,5)   (f) {3};
        \node[draw,dot,dashed] at (6,6)   (g) {4};
        \node[draw,dot,very thick] at (7,7)   (h) {5};
        \node[draw,dot] at (3,1)   (l1) {};
        \node[draw,dot] at (6,4)   (l2) {};
        \node[draw,dot] at (8,6)   (l3) {};
    \end{tikzpicture}
    \caption{An example of a binary spinal tree-child network with interior vertices labelled according to the tree-child labelling scheme described in Theorem~\ref{t:binspinTCN}. Each vertex is drawn according to which category they are in. Category $1$ are plain circles, Category $2$ are bold, and Category $3$ are dashed.} 
    \label{fig:tcn}
\end{figure}
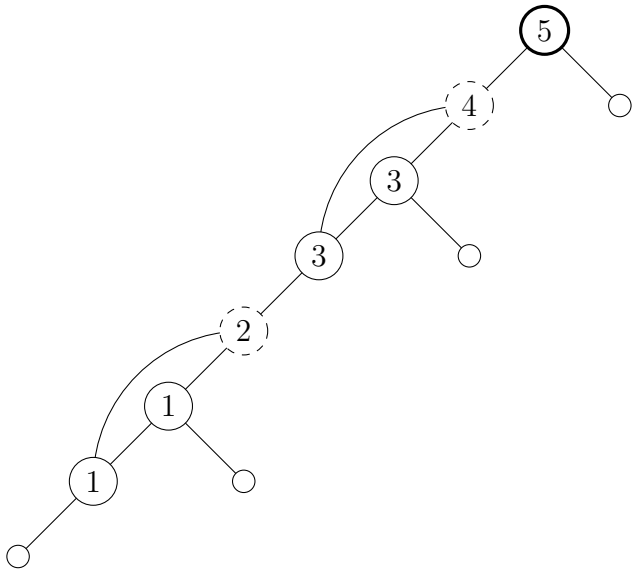

\begin{Example}
    Take the network shape depicted in Figure \ref{fig:tcn}. In this case we have a non-spinal edge connecting a Category $1$ vertex labelled $1$ to a Category $3$ vertex labelled $2$. Similarly, the label $3$ is connected to the label $4$. Finally, there is a single Category $2$ vertex, labelled $5$. Therefore, the partition this network shape is mapped to is $12|34|5$.
\end{Example}

The sequence $|\mathcal{STC}_{n,k}|$ does not appear to be in the OEIS, so we have included some small values in Table \ref{tab:tcncount}. However, as $k \le n-1$ due to $\mathcal{STC}_{n,k}$ consisting of tree-child networks, we can find further associated sequences. In particular, we have the following results.

\begin{Theorem}
Let $\mathcal{STC}_n$ be the set of binary spinal TCNs on $n$ leaves, and let $\mathcal{STCS}_n$ be the set of shapes of binary spinal TCNs on $n$ leaves. Then

\begin{equation}\label{e:STCn.count}
|\mathcal{STC}_n| = \sum_{k=0}^{n-1} |\mathcal{STC}_{n,k}|,    
\end{equation}
and
\[|\mathcal{STCS}_n| = \sum_{k=0}^{n-1} |\mathcal{B}_{n-1,k}|.\]

\end{Theorem}

The latter of these two sequences appears in the OEIS as A001515 \cite{oeis}, while the former does not appear to have a sequence in the OEIS. The first elements of the sequence $|\mathcal{STC}_n|$ can be found in the Total row in Table \ref{tab:tcncount}.

\begin{table}[h]
    \centering

\begin{tabular}{|l|c|c|c|c|c|c|c|c|}
    \hline
    \textbf{$k\backslash n$} & \textbf{1} & \textbf{2} & \textbf{3} & \textbf{4} & \textbf{5} & \textbf{6} & \textbf{7} & \textbf{8} \\ \hline
    \textbf{0} & 1 & 1 & 3 & 12 & 60 & 360 & 2520 & 20160 \\ \hline
    \textbf{1} &  & 2 & 15 & 108 & 840 & 7200 & 68040 & 705600 \\ \hline
    \textbf{2} &  &  & 18 & 324 & 4500 & 59400 & 793800 & 11007360 \\ \hline
    \textbf{3} &  &  &  & 360 & 11700 & 264600 & 5292000 & 101606400 \\ \hline
    \textbf{4} &  &  &  &  & 12600 & 642600 & 21432600 & 603288000 \\ \hline
    \textbf{5} &  &  &  &  &  & 680400 & 50009400 & 2305195200 \\ \hline
    \textbf{6} &  &  &  &  &  &  & 52390800 & 5239080000 \\ \hline
    \textbf{7} &  &  &  &  &  &  &  & 5448643200 \\ \hline
\textbf{Total} & 1 & 3 & 36 & 804 & 29700 & 1654560 & 129989160 & 13709545920 \\ \hline
\end{tabular}
    \caption{Values for $|\mathcal{STC}_{n,k}|$ where $1 \le n \le 8$ and $0 \le k < n$.}
    \label{tab:tcncount}
\end{table}

We also make the following observation regarding caterpillar trees.

\begin{Observation}    
The set of binary spinal tree-child networks on $n$ leaves with $0$ reticulations is exactly the set of caterpillar trees on $n$ leaves. By setting $k=0$ in Equation \eqref{e:bessel}, we see that $B_{n,0} =1$ for all $n$, and so can, by Theorem \ref{t:binspinTCN} (or by directly substituting into Equation \eqref{e:STCnk.countdirect}), retrieve the well-known fact that the number of caterpillar trees on $n$ leaves is exactly $n!/2$.
\end{Observation}

\subsection{The class of binary spinal stack-free networks}
\label{s:stackfree}

The enumeration of binary spinal stack-free networks proceeds in a similar manner to that for binary spinal tree-child networks. 

\begin{Definition}
    A network is \emph{stack-free} if every child of a reticulation vertex is a tree vertex.
\end{Definition}

Examples for $n=2$ leaves and $k=2$ reticulations are shown in Figure~\ref{f:binary.spinal.2.2}.
We begin with an analogous result to Theorem \ref{t:spinaltcn.RLP}, which is proved in a very similar way and so the proof is omitted.

\begin{Theorem}
\label{t:spinalsfn.covers}
Let $N$ be a binary spinal stack-free network with expanding cover $\C$. Then a set $C_i\in\C$ falls into one of the following four mutually exclusive categories:

\begin{enumerate}
    \item It is part of a consecutive pair of sets $C_i,C_{i+1}$ such that $C_i$ corresponds to an $R$-vertex and $C_{i+1}$ corresponds to  an $L$-vertex; or
    \item It is part of a consecutive pair of sets $C_i,C_{i+1}$ such that $C_i$ corresponds to an $R$-vertex and $C_{i+1}$ corresponds to  a $P$-vertex that has an edge to the first element of a pair in Categories $1$ or $2$ that is not $C_i$; or
    \item $C_i$ corresponds to an $L$-vertex and $C_{i-1}$ is not an $R$-vertex; or
    \item $C_i$ corresponds to a $P$-vertex that has an edge to the first element of a pair in Categories $1$ or $2$ and $C_{i-1}$ is not an $R$-vertex.
\end{enumerate}

Indeed, for each pair of sets in Category $1$, there is a corresponding $P$-vertex in Category $2$ or $4$.
\end{Theorem}

\begin{Observation}
    Considering the rules of the finite sequences in the proof sketch in Section \ref{s:proofsketch}, this is equivalent to adding a new rule that is a relaxed form of the rule for tree-child networks: \textbf{(R4b)} $R_i$ may be followed immediately by $L$ or $P_j$ for $i \ne j$, but cannot be followed by any $R_j$. Our enumeration proof will therefore count how many ways to arrange subsequences of the forms $R_iL,R_iP_j,L,P_i$ in which, for some $0<m<k$ there are $m$ of the form $R_iL$, $k-m$ of the form $R_iP_j$, $n-1-m$ of the form $L$ and $m$ of the form $P_i$, to make our sequence, while still enforcing the rules mentioned in Section \ref{s:proofsketch}. This will be achieved by mapping to an appropriate set.
\end{Observation}

Denote the set of partitions of a set of $n$ distinct elements into $k$ parts by $S_2(n,k)$, noting that these are counted by the Stirling numbers of the second kind. A closed formula for these exists, namely

\begin{equation}\label{e:S2.count}
|S_2(n,k)| = \sum_{j=0}^k \frac{(-1)^{k-j}j^n}{(k-j)!j!}.  
\end{equation}

We can now show the main result of this section.
Let $\mathcal{SSF}_{n,k}$ be the class of binary spinal stack-free networks with $n$ leaves and $k$ reticulations. Then we have the following.
\begin{Theorem}\label{t:spinalsfn}
\[|\mathcal{SSF}_{n,k}|=n!\left|S_2(n-1+k,n-1)\right|-\frac{n!}{2}|S_2(n-2+k,n-2)|.\]

Equivalently,

\begin{equation}
    |\mathcal{SSF}_{n,k}|=n(n-1)^{n-1+k} - \frac{n!}{2} \sum_{j=0}^{n-2} \frac{(-1)^{n-j}(n-1+j)j^{n-2+k}}{(n-1-j)!j!}.
\end{equation}
\end{Theorem}

\begin{Observation}
As with the case of tree-child networks, note that when $k=0$, these expressions reduce to $\frac{n!}{2}$, which counts the number of caterpillar trees on $n$ leaves (this is easiest to see in the first expression, because $S_2(n-1,n-1)=1$).
\end{Observation}

\begin{proof}
    The enumeration of binary spinal stack-free networks  proceeds similarly to the TCN case. Again we label vertices along the path of the spine (ignoring the initial leaf vertex), but this time repeating a label for the consecutive pairs in Categories $1$ \emph{and} $2$ (Theorem~\ref{t:spinalsfn.covers}). However, in this case we now construct equivalence classes of labels as follows. Define the equivalence $\sim$ by setting:
\begin{enumerate}
    \item $i \sim i$ for all labels $i$, and
    \item $i \sim j$ if $i$ labels an $R$-vertex with a corresponding $P$-vertex labelled by $j$.
\end{enumerate}

 We claim that the equivalence relation induces a surjection from $\mathcal{SSF}_{n,k}$ to $S_2(n-1+k,n-1)$. We note first that there are $n-1+2k$ internal vertices of a given element of $\mathcal{SSF}_{n,k}$, of which $k$ are reticulation vertices, which fall into Categories $1$ and $2$, and hence there are $n-1+k$ distinct labels.

Denote the number of (pairs of) sets in Categories $1$ and $3$ by $a$ and $b$ respectively. 
Each leaf appears only once per element of a Category, and so the $n-1$ leaves not on the spine can only appear in Categories $1$ and $3$; thus $a+b = n-1$. 
Similarly, the least element in an equivalence relation must be in Categories $1$ or $3$, so the number of  parts induced by the equivalence relation is also $a+b$. It follows that there will be $n-1$ parts induced by the equivalence relation. It is clear the partition induced by the equivalence class is unique for each network up to permutation of the leaves, and hence the number of network shapes in $\mathcal{SSF}_{n,k}$ is $|S_2(n-1+k,n-1)|$, which is a sequence appearing in the OEIS (off by 1) as A354977 \cite{oeis}. By Lemma \ref{l:shapetonet}, we obtain the first statement of the theorem. The closed form is then obtained via Equation \eqref{e:S2.count} and some algebraic manipulation.
\end{proof}

\begin{Example}
    Take the network shape depicted in Figure \ref{fig:sfn}. In this case we have a non-spinal edge connecting a pair of Category $1$ vertices labelled $1$ to a Category $2$ vertex labelled $2$, which are themselves connected to a Category $4$ vertex labelled $4$. Hence, by the equivalence relation, $1 \sim 2 \sim 4$. However, the label of the Category $3$ vertex is only equivalent to itself. Hence, this network shape maps to $124|3$ in the stack-free labelling scheme.
\end{Example}

\begin{figure}
    \centering
    \begin{tikzpicture}[dot/.style={circle,fill=white,radius=3pt,inner sep = 3pt}] 
        \draw (0, 0) -- (6,6);
        \draw (2, 2) -- (3,1);
        \draw (5,5) -- (6,4);
        \draw [bend left=45] (1,1) to (4,4);
        \draw [bend left=45] (3,3) to (6,6);
        
        \node[draw,dot] at (0, 0)   (a) {};
        \node[draw,dot] at (1, 1)   {1};
        \node[draw,dot] at (2, 2)   (c) {1};
        \node[draw,dot,very thick] at (3, 3)   (d) {2};
        \node[draw,dot,very thick] at (4, 4)   (e) {2};
        \node[draw,dot,dashed] at (5,5)   (f) {3};
        \node[draw,dot,dotted] at (6,6)   (g) {4};
        \node[draw,dot] at (3,1)   (l1) {};
        \node[draw,dot] at (6,4)   (l2) {};
    \end{tikzpicture}
    \caption{An example of a binary spinal stack-free network with interior vertices labelled according to the stack-free labelling scheme described in the proof of Theorem \ref{t:spinalsfn}. Each vertex is drawn according to which category they are in. Category $1$ are plain circles, Category $2$ are dashed, Category $3$ are bold, and Category $4$ are dotted. This example is mapped to the partition $124|3$.}
    \label{fig:sfn}
\end{figure}

This sequence does not appear to be in the OEIS, so we have included some small values in Table \ref{tab:ssfcount}. As binary spinal stack-free networks do not have a limit on the number of reticulations that may appear, there are infinitely many for each $n$, and so we cannot achieve an analogous result to the previous section where we found an expression for $|\mathcal{STC}_n|$ by summing over all $k$.

\begin{table}[h]
    \centering
\begin{tabular}{|l|c|c|c|c|c|c|c|}
    \hline
    \textbf{$k\backslash n$} & \textbf{2} & \textbf{3} & \textbf{4} & \textbf{5} & \textbf{6} & \textbf{7} & \textbf{8} \\ \hline
    \textbf{0} & 1 & 3 & 12 & 60 & 360 & 2520 & 20160 \\ \hline
    \textbf{1} & 2 & 15 & 108 & 840 & 7200 & 68040 & 705600 \\ \hline
    \textbf{2} & 2 & 39 & 516 & 6300 & 77400 & 987840 & 13265280 \\ \hline
    \textbf{3} & 2 & 87 & 1980 & 36600 & 630000 & 10689840 & 183738240 \\ \hline
    \textbf{4} & 2 & 183 & 6852 & 186060 & 4392360 & 97531560 & 2119763520 \\ \hline
    \textbf{5} & 2 & 375 & 22428 & 874440 & 27820800 & 797451480 & 21678148800 \\ \hline
    \textbf{6} & 2 & 759 & 71076 & 3911100 & 165367800 & 6049446480 & 203761071360 \\ \hline
    \textbf{7} & 2 & 1527 & 220860 & 16930200 & 940698000 & 43503324480 & 1801038919680 \\ \hline
    \textbf{8} & 2 & 3063 & 677892 & 71670060 & 5185980360 & 300797897400 & 15201573266880 \\ \hline
\end{tabular}
    \caption{The number of binary spinal stack-free networks $|\mathcal{SSF}_{n,k}|$, for $2 \le n \le 8$ and $0 \le k \le 8$, obtained using Theorem~\ref{t:spinalsfn}. }
    \label{tab:ssfcount}
\end{table}

\subsection{The class of binary spinal fully tree-sibling networks}
\label{s:fullytreesibling}

The tree-child result in Equation~\eqref{e:STCnk.count} can be extended to ``fully tree-sibling'' networks.

\begin{Definition}
    A network $N$ is \emph{fully tree-sibling} if and only if every sibling of a reticulation vertex is a tree vertex.
\end{Definition}

The class of fully tree-sibling networks, denoted $\mathcal{FTS}_{n,k}$ generalises the class of tree-child networks, as stacks are permitted, but are a subclass of tree-sibling networks, as tree-sibling networks only require each reticulation vertex to have at least one sibling that is a tree vertex.  Some examples are shown in Figure~\ref{f:binary.spinal.2.2}.

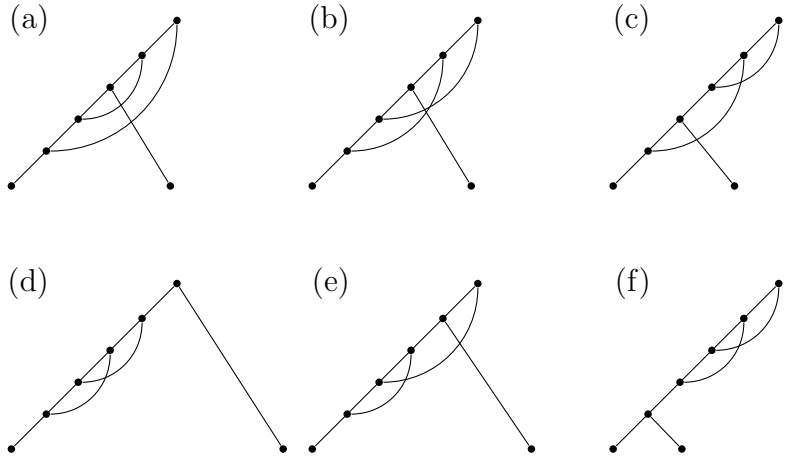
\begin{figure}[ht]
\centering
\begin{tikzpicture}[ 
     sdot/.style={circle,fill,radius=1pt,inner sep=1pt}] 

\node[sdot] (1) {};  
\node[sdot, above right=3cm of 1] (r) {};
\node[sdot,right=20mm of 1] (2) {};
\draw (r) to 
     node[pos=.2,sdot] (a) {} 
     node[pos=.4,sdot] (b) {} 
     node[pos=.6,sdot] (c) {} 
     node[pos=.8,sdot] (d) {} (1);
\draw (b)--(2);
\draw (r) to[out=-90,in=0] (d);
\draw (a) to[out=-90,in=0] (c);
\node[left=15mm of r] () {(a)};

\begin{scope}[xshift=40mm]
\node[sdot] (1) {};  
\node[sdot, above right=3cm of 1] (r) {};
\node[sdot,right=20mm of 1] (2) {};
\draw (r) to 
     node[pos=.2,sdot] (a) {} 
     node[pos=.4,sdot] (b) {} 
     node[pos=.6,sdot] (c) {} 
     node[pos=.8,sdot] (d) {} (1);
\draw (b)--(2);
\draw (r) to[out=-90,in=0] (c);
\draw (a) to[out=-90,in=0] (d);
\node[left=15mm of r] () {(b)};
\end{scope}

\begin{scope}[xshift=80mm]
\node[sdot] (1) {};  
\node[sdot, above right=3cm of 1] (r) {};
\node[sdot,right=15mm of 1] (2) {};
\draw (r) to 
     node[pos=.2,sdot] (a) {} 
     node[pos=.4,sdot] (b) {} 
     node[pos=.6,sdot] (c) {} 
     node[pos=.8,sdot] (d) {} (1);
\draw (c)--(2);
\draw (r) to[out=-90,in=0] (b);
\draw (a) to[out=-90,in=0] (d);
\node[left=15mm of r] () {(c)};
\end{scope}

\begin{scope}[yshift=-35mm]
\node[sdot] (1) {};  
\node[sdot, above right=3cm of 1] (r) {};
\node[sdot,right=35mm of 1] (2) {};
\draw (r) to 
     node[pos=.2,sdot] (a) {} 
     node[pos=.4,sdot] (b) {} 
     node[pos=.6,sdot] (c) {} 
     node[pos=.8,sdot] (d) {} (1);
\draw (r)--(2);
\draw (a) to[out=-90,in=0] (c);
\draw (b) to[out=-90,in=0] (d);
\node[left=15mm of r] () {(d)};
\end{scope}

\begin{scope}[yshift=-35mm,xshift=40mm]
\node[sdot] (1) {};  
\node[sdot, above right=3cm of 1] (r) {};
\node[sdot,right=28mm of 1] (2) {};
\draw (r) to 
     node[pos=.2,sdot] (a) {} 
     node[pos=.4,sdot] (b) {} 
     node[pos=.6,sdot] (c) {} 
     node[pos=.8,sdot] (d) {} (1);
\draw (a)--(2);
\draw (r) to[out=-90,in=0] (c);
\draw (b) to[out=-90,in=0] (d);
\node[left=15mm of r] () {(e)};
\end{scope}

\begin{scope}[yshift=-35mm,xshift=80mm]
\node[sdot] (1) {};  
\node[sdot, above right=3cm of 1] (r) {};
\node[sdot,right=8mm of 1] (2) {};
\draw (r) to 
     node[pos=.2,sdot] (a) {} 
     node[pos=.4,sdot] (b) {} 
     node[pos=.6,sdot] (c) {} 
     node[pos=.8,sdot] (d) {} (1);
\draw (d)--(2);
\draw (r) to[out=-90,in=0] (b);
\draw (a) to[out=-90,in=0] (c);
\node[left=15mm of r] () {(f)};
\end{scope}
\end{tikzpicture}
\caption{All binary spinal network shapes with $n=2$ leaves and $k=2$ reticulations.  Networks (a) and (b) are fully tree-sibling but not stack-free, while (c) is stack-free but not fully tree-sibling. Networks (d), (e), and (f) are neither fully tree-sibling or stack-free.  Note, there are no tree-child networks with $n=k=2$.  There are two leaf-labelled copies of each of shapes (a) to (e), while (f) has only one.}
    \label{f:binary.spinal.2.2}
\end{figure}

Our intent is now to enumerate binary spinal fully tree-sibling networks. As the proof for this theorem is very similar to the previous proofs for the corresponding theorems in the previous sections, the proof is omitted.

\begin{Theorem}\label{t:spinalftsn.covers}
Let $N$ be a binary spinal fully tree-sibling network with expanding cover $\C$. Then a set $C_i\in\C$ falls into one of the following three mutually exclusive categories: 

\begin{enumerate}
    \item It is part of a consecutive sequence of sets $C_i,\dots,C_{h}$ such that $C_i,\dots,C_{h-1}$ correspond to $R$-vertices and $C_{h}$ corresponds to  an $L$-vertex; or
    \item $C_i$ corresponds to an $L$-vertex and $C_{i-1}$ is not an $R$-vertex; or
    \item $C_i$ corresponds to a $P$-vertex that has an edge to some $R$-vertex of a sequence in Category $1$.
\end{enumerate}

Indeed, for each sequence of sets of length $h-1$ in Category $1$, there are $h-1$ corresponding $P$-vertices in Category $3$.
\end{Theorem}

\begin{Observation}
    Considering the rules of the finite sequences in the proof sketch in Section \ref{s:proofsketch}, this is equivalent to adding a new rule that is a relaxed form of the rule for tree-child networks: \textbf{(R4c)} $R_i$ may be followed immediately by $L$ or $R_j$ for $i < j$, but cannot be followed by any $P_j$. Our enumeration proof will therefore count how many ways to arrange subsequences of the forms $R_i\dots R_hL,L,P_i$ in which, for some $0< m \le k$ there are $m$ of the form $R_i\dots R_hL$, $n-1-m$ of the form $L$ and $k$ of the form $P_i$, to make our sequence, while still enforcing the rules mentioned in Section \ref{s:proofsketch}. This will be achieved by mapping to an appropriate set.
\end{Observation}

Denote the set of permutations of a set of $n$ distinct elements with $k$ cycles by $S_1(n,k)$, noting that these are counted by the Stirling numbers of the first kind. We can now show the main result of this section.  Let $\mathcal{FTS}_{n,k}$ be the class of fully tree-sibling binary spinal networks on $n$ leaves with $k$ reticulations. Then we have the following.

\begin{Theorem}\label{t:spinalftsn}
\[\left\vert\mathcal{FTS}_{n,k}\right\vert=n!\left|S_1(n-1+k,n-1)\right|-\frac{n!}{2}|S_1(n-2+k,n-2)|.\]

\end{Theorem}

\begin{proof}
     This proof will proceed similarly to the corresponding enumeration results from the previous two sections. We construct a surjection from the set $\mathcal{FTS}_{n,k}$ to $S_1(n-1+k,n-1)$ in the following way. Again we label along the path of the spine (ignoring the initial leaf vertex), but this time repeating a label for the consecutive sequences in Category $1$ (Theorem~\ref{t:spinalftsn.covers}). Now, Category $2$ $L$-vertices are mapped to singleton cycles, corresponding to their label. A Category $1$ sequence, say $S$, together with its corresponding $h-1$ $P$-vertices are mapped to $h$-cycles in which the first element is the label of the vertices corresponding to $S$, and the remaining elements are the labels of the corresponding $P$-vertices, ordered according to the position of their corresponding $R$-vertex in $S$. 

     We first note that there will be $n-1+k$ labels, following from the facts that there are $n-1$ total labels assigned to Categories $1$ and $2$ (as there is exactly one leaf for each label corresponding to a set of sets or set of Category $1$ and $2$ respectively), as well as exactly $k$ elements in Category $3$, due to the bijective relationship between $P$-vertices and $R$-vertices in a network in $\mathcal{FTS}_{n,k}$. We now claim there will be $n-1$ such cycles. This follows directly from the fact that we produce a cycle exactly once for each set of sets in Category $1$ and each set in Category $2$, of which there are precisely $n-1$ total.

It is again clear that the mapped permutation is unique for each network up to permutation of the leaves, and hence the number of network shapes in $\mathcal{FTS}_{n,k}$ is $|S_1(n-1+k,n-1)|$, which appears in the OEIS (off by 1) as A354979 \cite{oeis}. By Lemma \ref{l:shapetonet}, we obtain the statement of the theorem.
\end{proof}

This sequence does not appear to be in the OEIS, so we have included some small values in Table \ref{tab:ftsncount}.

\begin{table}[h]
    \centering
        \resizebox{\textwidth}{!}{
\begin{tabular}{|l|c|c|c|c|c|c|c|}
    \hline
    \textbf{$k\backslash n$} & \textbf{2} & \textbf{3} & \textbf{4} & \textbf{5} & \textbf{6} & \textbf{7} & \textbf{8} \\ \hline
    \textbf{0} & 1 & 3 & 12 & 60 & 360 & 2520 & 20160 \\ \hline
    \textbf{1} & 2 & 15 & 108 & 840 & 7200 & 68040 & 705600 \\ \hline
    \textbf{2} & 4 & 60 & 708 & 8100 & 95400 & 1181880 & 15523200 \\ \hline
    \textbf{3} & 12 & 282 & 4800 & 74700 & 1146600 & 17922240 & 289578240 \\ \hline
    \textbf{4} & 48 & 1572 & 35688 & 714840 & 13726440 & 262324440 & 5085823680 \\ \hline
    \textbf{5} & 240 & 10224 & 294000 & 7286160 & 169691760 & 3867658200 & 88160909760 \\ \hline
    \textbf{6} & 1440 & 76248 & 2678160 & 79754160 & 2199664800 & 58620592800 & 1545081088320 \\ \hline
    \textbf{7} & 10080 & 642384 & 26829792 & 938778000 & 30089505600 & 922684316400 & 27736453704960 \\ \hline
    \textbf{8} & 80640 & 6038496 & 293766912 & 11865754560 & 435296014560 & 15157877454720 & 513727745731200 \\ \hline
\end{tabular}
    }
    \caption{The number of binary spinal fully tree-sibling networks $|\mathcal{FTS}_{n,k}|$, for $2 \le n \le 8$ and $0 \le k \le 8$, using Theorem~\ref{t:spinalftsn}. }
    \label{tab:ftsncount}
\end{table}

\begin{Example}
For example, in Figure \ref{fig:ftsn}, the blue vertices labelled $1$ are Category $1$ vertices, and the $R$-vertices in these have an edge to the non-spinal parents $3$ and $4$. As the $R$-vertex has an adjacent non-spinal vertex labelled $4$, and the second has an edge to $3$, we obtain the cycle $(143)$. As the (only) Category $2$ vertex is labelled $2$, we obtain the cycle $(2)$. Putting this together, this network shape corresponds to the permutation $(143)(2)$.
\end{Example}

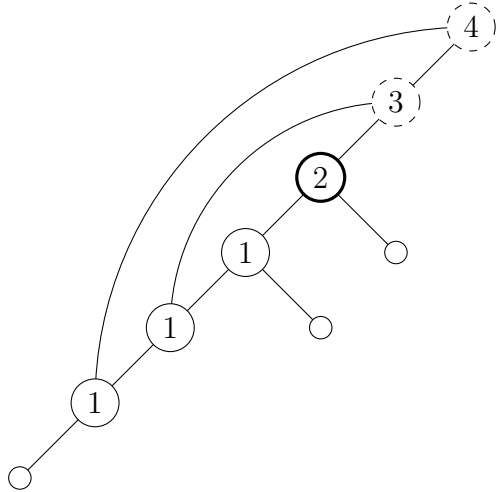
\begin{figure}[ht]
    \centering
    \begin{tikzpicture}[dot/.style={circle,fill=white,radius=3pt,inner sep = 3pt}] 
        \draw (0, 0) -- (6,6);
        \draw (3, 3) -- (4,2);
        \draw (4,4) -- (5,3);
        \draw [bend left=45] (1,1) to (6,6);
        \draw [bend left=45] (2,2) to (5,5);
        
        \node[draw,dot] at (0, 0)   (a) {};
        \node[draw,dot] at (1, 1)   {1};
        \node[draw,dot] at (2, 2)   (c) {1};
        \node[draw,dot] at (3, 3)   (d) {1};
        \node[draw,dot,very thick] at (4, 4)   (e) {2};
        \node[draw,dot,dashed] at (5,5)   (f) {3};
        \node[draw,dot,dashed] at (6,6)   (g) {4};
        \node[draw,dot] at (4,2)   (l1) {};
        \node[draw,dot] at (5,3)   (l2) {};
    \end{tikzpicture}
    \caption{An example of a binary spinal fully tree-sibling network with interior vertices labelled according to the tree-sibling labelling scheme. Each vertex is drawn according to which category they are in. Category $1$ are plain circles, Category $2$ are dashed, and Category $3$ are bold. This example is mapped to, in cycle notation, $(143)(2)$.}
    \label{fig:ftsn}
\end{figure}

As binary spinal fully tree-sibling networks do not have a limit on the number of reticulations that may appear, there are infinitely many for each $n$, and so we cannot achieve an analogous result to Section \ref{s:caterpillar} where we found an expression for $|\mathcal{STC}_n|$ summing over all $k$.

\subsection{Other classes}

\subsubsection{Parallel edges allowed}
\label{s:parallel}

Consider the class of binary spinal networks on $n$ leaves with $k$ reticulations with no further restrictions, and allow parallel edges. Then, up to rearrangement of the leaves, this will be equivalent to labelling the internal vertices from $1$ to $n-1+2k$, and then selecting $k$ pairs of internal vertices and for each pair, say $i$ and $j$ with $i<j$, set the vertex labelled $i$ to be the non-spinal child of the vertex labelled $j$. Enumerating the class is therefore simply the number of ways of selecting $k$ pairs of numbers from $1$ to $n-1+2k$, which is

\[\frac{\binom{n-1+2k}{2}\binom{n-3+2k}{2} \dots \binom{n+1}{2}}{k!} = \frac{(n-1+2k)!}{2^k(n-1)!k!}.\]

Then, by Lemma \ref{l:shapetonet}, we have the following result.

\begin{Theorem}
    Let $\mathcal{BSPN}_{n,k}$ be the class of binary spinal networks, in which parallel edges are permitted. Then
    
\begin{equation} 
\begin{split}
|\mathcal{BSPN}_{n,k}| & =n!\left(\frac{(n-1+2k)!}{2^k(n-1)!k!}-\frac{(n-2+2k)!}{2^{k+1}(n-2)!k!}\right)\\
 & = \frac{n(n-1+4k)(n-2+2k)!}{2^{k+1}k!}.
\end{split}
\end{equation}
\end{Theorem}

\subsubsection{Parallel edges not allowed}

Enumeration of binary spinal networks on $n$ leaves with $k$ reticulations and no parallel edges, denoted $\mathcal{BSN}_{n,k}$, is similar to the previous case, except that we may not select a consecutive pair of vertices. Krasko and Omelchenko provide a recurrence relation in \cite[Lemma 2.2]{loopless} for a related problem that obtains the same sequence, involving counting linear diagrams obtained from chordal diagrams in which the number of chords between consecutive pairs is equal to a given value. Rather than construct a bijection between their construction and ours, we provide a direct proof that is adapted from theirs. This will count the number of network shapes, and a full count of networks may then be obtained using Lemma \ref{l:shapetonet}. Let $\mathcal{BSNS}_{n,k}$ denote the set of shapes of binary spinal networks on $n$ leaves with $k$ reticulations and no parallel edges. Additionally, let $\mathcal{BSN}_{n,k}$ denote the set of binary spinal networks on $n$ leaves with $k$ reticulations and no parallel edges. 

\begin{Theorem}
\label{t:binaryspin}
    \[|\mathcal{BSNS}_{n,k}| = |\mathcal{BSNS}_{n-1,k}| + (n+2k-3)|\mathcal{BSNS}_{n,k-1}| + n|\mathcal{BSNS}_{n+1,k-2}|, \]
    where $|\mathcal{BSNS}_{n,k}|=0$ for $k<0$, $|\mathcal{BSNS}_{n,0}|=1$ for $n \ge 1$, and $|\mathcal{BSNS}_{0,k}|=0$.

    Additionally, 

    \[|\mathcal{BSN}_{n,k}|=n!|\mathcal{BSNS}_{n,k}|-\frac{n!}{2}|\mathcal{BSNS}_{n-1,k}|.\]
    
\end{Theorem}

\begin{proof}
 We consider base cases first. There are no networks with zero leaves and no networks with a negative number of reticulations, so $|\mathcal{BSNS}_{0,k}|=0$ and $|\mathcal{BSNS}_{n,k}|=0$ for $k<0$. Clearly the only network shape in $\mathcal{BSNS}_{n,0}$ is the caterpillar tree, so $|\mathcal{BSNS}_{n,0}|=1$. 

    Consider a tree shape in $\mathcal{BSNS}_{n,k}$. There are three possibilities for the final vertex $v$ (the $n+2k$-th):
    \begin{enumerate}
        \item The non-spinal adjacent vertex of $v$ is a leaf; or
        \item The non-spinal adjacent vertex of $v$ is a reticulation vertex $w$, and deletion of $v$ and suppression of $w$ would not result in a parallel edge; or
        \item The non-spinal adjacent vertex of $v$ is a reticulation vertex $w$, and deletion of $v$ and suppression of $w$ would result in a parallel edge.
    \end{enumerate}
    We can thus construct the network shapes in $\mathcal{BSNS}_{n,k}$ in the following three mutually exclusive ways:
    \begin{enumerate}
        \item Take a network shape in $\mathcal{BSNS}_{n-1,k}$, add a new (root) vertex at the end of the spinal path, with a new leaf for its non-spinal adjacent vertex. For each given network shape, there is only one way to do this.
        \item Take a network shape in $\mathcal{BSNS}_{n,k-1}$, pick an edge of the path and subtend it to form a new vertex $w$. Then add a new (root) vertex $v$ at the end of the path, and an edge from $v$ to $w$. There are $n+2k-3$ edges on the path, so $n+2k-3$ options for this.
        \item Take a network shape in $\mathcal{BSNS}_{n+1,k-2}$, and choose a vertex $w$ for which the non-spinal adjacent vertex is one of the $n$ leaves that are not on the spinal path. Delete this leaf, add a new (root) vertex $v$ at the end of the path, and add an edge from $v$ to $w$. Finally, subtend the two spinal edges adjacent to $w$, and add an edge between the two new vertices created. There are $n$ options for this.
    \end{enumerate}

    From these observations, the first statement of the theorem follows. We obtain the final statement by application of Lemma \ref{l:shapetonet}.
\end{proof}

Neither of the corresponding sequences from Theorem \ref{t:binaryspin} appear to be in the OEIS, so we have included some small values in Tables \ref{tab:binarycount} and \ref{tab:binarycountfull}. Note in particular that the only zero value in either table outside of $n=0$ is the case $n=1,k=1$. In this case, one can see that no such network can exist, as a binary spinal network with one leaf and one reticulation has precisely $3$ vertices, one of which is a leaf, so the reticulation must be the parent of the leaf and the parent of the reticulation along the spine must also be the parent of the reticulation along the non-spinal edge, contradicting the fact that we are not allowing parallel edges. 

\begin{table}[h]
    \centering
        \resizebox{\textwidth}{!}{
\begin{tabular}{|c|c|c|c|c|c|c|c|c|}
    \hline
    \textbf{k \textbackslash n} & \textbf{0} & \textbf{1} & \textbf{2} & \textbf{3} & \textbf{4} & \textbf{5} & \textbf{6} & \textbf{7} \\ \hline
    \textbf{0} & 0  & 1  & 1   & 1    & 1      & 1       & 1         & 1  \\ \hline
    \textbf{1} & 0  & 0  & 1   & 3    & 6      & 10      & 15        & 21  \\ \hline
    \textbf{2} & 0  & 1  & 6   & 21   & 55     & 120     & 231       & 406  \\ \hline
    \textbf{3} & 0  & 5  & 41  & 185  & 610    & 1645    & 3850      & 8106  \\ \hline
    \textbf{4} & 0  & 36  & 365 & 2010 & 7980   & 25585   & 70371     & 172305  \\ \hline
    \textbf{5} & 0  & 329  & 3984 & 25914 & 120274  & 446544  & 1410003    & 3932313  \\ \hline
    \textbf{6} & 0  & 3655  & 51499 & 386407 & 2052309  & 8655780  & 30839655   & 96451278  \\ \hline
    \textbf{7} & 0  & 47844  & 769159 & 6539679 & 39110490  & 184652985  & 732520998   & 2538187938 \\ \hline
    \textbf{8} & 0  & 721315  & 13031514 & 123823305 & 823324755  & 4301276760  & 18797883390  & 71463760368   \\ \hline
\end{tabular}
    }
    \caption{The number of binary spinal network shapes with no parallel edges $|\mathcal{BSNS}_{n,k}|$, for $0 \le n \le 7$, $0 \le k \le 8$, using Theorem~\ref{t:binaryspin}.}
    \label{tab:binarycount}
\end{table}

\begin{table}[h]
    \centering
        \resizebox{\textwidth}{!}{
\begin{tabular}{|l|c|c|c|c|c|c|c|c|}
    \hline
    \textbf{k \textbackslash n} & \textbf{0} & \textbf{1} & \textbf{2} & \textbf{3} & \textbf{4} & \textbf{5} & \textbf{6} & \textbf{7} \\ \hline
    \textbf{0} & 0 & 1 & 1 & 3 & 12 & 60 & 360 & 2520 \\ \hline
    \textbf{1} & 0 & 0 & 2 & 15 & 108 & 840 & 7200 & 68040 \\ \hline
    \textbf{2} & 0 & 1 & 11 & 108 & 1068 & 11100 & 123120 & 1464120 \\ \hline
    \textbf{3} & 0 & 5 & 77 & 987 & 12420 & 160800 & 2179800 & 31152240 \\ \hline
    \textbf{4} & 0 & 36 & 694 & 10965 & 167400 & 2591400 & 41456520 & 691082280 \\ \hline
    \textbf{5} & 0 & 329 & 7639 & 143532 & 2575608 & 46368840 & 854446320 & 16265649960 \\ \hline
    \textbf{6} & 0 & 3655 & 99343 & 2163945 & 44618532 & 915555060 & 19088470800 & 408398510520 \\ \hline
    \textbf{7} & 0 & 47844 & 1490474 & 36930597 & 860175612 & 19811728800 & 460940043960 & 10946514292560 \\ \hline
    \textbf{8} & 0 & 721315 & 25341713 & 703845288 & 18273914460 & 466753725900 & 11986016407200 & 312806686111920 \\ \hline
\end{tabular}
    }
    \caption{The number of binary spinal networks with no parallel edges $|\mathcal{BSN}_{n,k}|$, for $0 \le n \le 7$, $0 \le k \le 8$ using Theorem~\ref{t:binaryspin}.}
    \label{tab:binarycountfull}
\end{table}

\section{Discussion}

In this paper we have enumerated several classes of phylogenetic network that generalise the class of phylogenetic trees known as ``caterpillar'' trees.  These trees are important mathematically, because they sit at one extreme of the space of rooted trees, as the least ``balanced'', and their combinatorial properties are easier to handle.  They are also of interest biologically, as in the presence of rapid evolution, such as with viral pathogens, the trees that arise are often close to caterpillar in shape~\cite{fitch1991positive,bush2000effects,grenfell2004unifying}.

Phylogenetic networks pose many challenges with enumeration, because of their complexity, and so structural restrictions that provide the possibility of concrete results can be important.   
The present paper takes the approach of applying heavy restrictions on network structure, while taking advantage of a new framework for set-theoretic bijections with networks that have been recently developed~\cite{francis2023phylogenetic,francis2023labellable}.

The use of covers for labellable networks may also help in finding recursive formulae for networks along the lines of the Pons-Batle conjecture in Eq.~\eqref{e:TCN.recursion}~\cite{pons2021combinatorial}, and this seems an interesting direction for future research.

\section*{Acknowledgements}
AF would like to thank the Institute for Computational and Experimental Research in Mathematics (ICERM) in Providence, Rhode Island, for its hospitality during some of the writing of this paper.


\bibliographystyle{plain}
\bibliography{sample}

\end{document}